\documentclass[12pt]{article}

\usepackage{epsfig}

\usepackage{amssymb}
\usepackage{amsfonts}

\usepackage{color}
 
\def\be{\begin{equation}}
\def\ee{\end{equation}}
\def\ba{\begin{array}{c}}
\def\ea{\end{array}}

\newcommand{\bea}{\begin{eqnarray}}
\newcommand{\eea}{\end{eqnarray}}

\newcommand{\kkt}{\kt\!\kt}

\newcommand{\pkt}{\!\!\succ\,}
\newcommand{\kt}{\rangle}
\newcommand{\br}{\langle}

\newtheorem{thm}{Theorem}

\newtheorem{lemma}[thm]{Lemma}

\newenvironment{proof}{\noindent
 {\bf Proof.}}{\hfill$\square$\vspace{3mm}\endtrivlist}

\begin{document}


\vspace{.35cm}

\begin{center}

{\Large

Theory of
response to perturbations in non-Hermitian systems
  using
five-Hilbert-space reformulation of unitary quantum mechanics

 }

\vspace{10mm}

\textbf{Miloslav Znojil}

\vspace{5mm}

Institute of System Science, Durban University of Technology,
P. O. Box 1334, Durban, 4000, South Africa

and

The Czech Academy of Sciences, Nuclear Physics Institute,
 Hlavn\'{\i} 130,
250 68 \v{R}e\v{z}, Czech Republic

%
%
%
%

\end{center}



\section*{Abstract}

Non-Hermitian quantum-Hamiltonian-candidate combination $H_\lambda$
of a non-Hermitian unperturbed operator $H=H_0$ with
an arbitrary ``small'' non-Hermitian perturbation $\lambda W$ is given a
mathematically consistent unitary-evolution interpretation.
The formalism generalizes the conventional
constructive Rayleigh-Schr\"{o}dinger
perturbation expansion technique. It
is sufficiently general to take into account
the well known
formal ambiguity of the reconstruction of the correct
physical Hilbert space of states. The possibility of removal
of the ambiguity
via a complete, irreducible set of
observables is also discussed.

\subsection*{Keywords}.

  hidden Hermiticity;

  Hilbert space metric;

  size of perturbations;

  stability;

  PT symmetry;

  unitary quantum evolution;

%


\newpage

\section{Introduction}

``How much would a small change of a Hamiltonian influence the
measurable properties of a unitary quantum system?" This is the
question answered, in general terms, by perturbation theory
\cite{Kato}. In a less abstract setting it makes sense to
distinguish between the two alternative scenarios,
viz., Hermitian theory and
non-Hermitian theory. In our present paper we
intend to study a few less trivial questions emerging in the latter
context.

We feel motivated by the observation that the conventional textbooks
on quantum mechanics rarely leave the former, Hermitian-theory
framework \cite{Messiah}. For our present purposes we will call this
framework the ``single Hilbert space (1HS) formulation of quantum
mechanics'' because, typically,
the states are represented in a fixed, preselected Hilbert space
(say, ${\cal L}$). In the
so called Schr\"{o}dinger picture, the evolution of the wave
functions (denoted by the curly-ket symbol
$|\bullet\pkt$ in what follows) is then required controlled by
a suitable self-adjoint Hamiltonian (say,
$\mathfrak{h}=\mathfrak{h}^\dagger$), i.e., by
Schr\"{o}dinger equation
 \be
 {\rm i}\frac{d}{dt} |\psi(t)\pkt \,
 = \,\mathfrak{h}\,|\psi(t)\pkt\,,
 \ \ \ \ \ |\psi(t)\pkt \, \in \, {\cal L}\,.
 \label{tse}
 \ee
In the stationary
bound-state models the construction of the solutions then
degenerates to the diagonalization,
 \be
 \mathfrak{h}\,|\psi^{(n)}\pkt =
 E^{(n)}\,|\psi^{(n)}\pkt\,,
  \ \ \ \
 |\psi^{(n)}\pkt \ \in \ {\cal L}\,,
 \ \ \ \ \ \ \
 n = 0, 1, \ldots\,.
 \label{selfadjo}
 \ee
Also the study of influence of perturbations is being performed,
predominantly, in the same 1HS framework. The reference is made to
the well known Stone theorem \cite{Stone} and the
self-adjointness is postulated also for the perturbations,
 \be
 \mathfrak{h}\ \to \
 \mathfrak{h}_\lambda=\mathfrak{h}+ \lambda
 \mathfrak{v}=\mathfrak{h}_\lambda^\dagger\,.
 \label{conve}
 \ee
Any non-Hermiticity is excluded as forbidden. The
perturbation-approximation estimates of energies $E_\lambda^{(n)}$
and/or of state vectors $|\psi_\lambda^{(n)} \pkt\, \in\, {\cal L}$
are then calculated using, e.g., the
Rayleigh-Schr\"{o}dinger power-series ansatzs \cite{Messiah}.

By Bender with Boettcher \cite{BB}
the apparently obvious constraint of self-adjointness (\ref{conve}) has
been claimed over-restrictive. In
sections \ref{prepu} and \ref{realth}, the
subsequent consistent amendment of the formalism will be
reviewed and summarized
(cf. also
the older references \cite{Dyson,Geyer}
and newer reviews \cite{Carl,ali,book,Ray} in this respect).
In this preparatory text we will summarize the
current state of art, and its abstract outline will be
complemented by a
schematic, ``relativistic'' illustrative example.

It is worth adding that in the current literature
such a version of the unitary
quantum theory
became widely known
under
the most popular names of the ````non-Hermitian but ${\cal
PT}-$symmetric'' \cite{Carl} {\it alias\,}
``pseudo-Hermitian'' \cite{ali} quantum
mechanics.
More explanatory than ``non-Hermitian'' would be, certainly,
``quasi-Hermitian'' \cite{Geyer,ali} or
``crypto-Hermitian'' \cite{Smilga}.
A few other comments on
the not yet unified
terminology may be also found in our papers \cite{NIP,corridors}.
Here, in a way proposed in
\cite{SIGMA} the
corresponding innovative but still strictly unitary amendment of
Schr\"{o}dinger picture
will be simply called
the ``three Hilbert space
(3HS) formulation of quantum
mechanics''.

In this framework
our key methodical observation
is that
the underlying
mathematics is
comparatively well understood only
under a tacit assumption that there are no
perturbations
which could enter the game.
In this situation
the quantum physicists
communicating with experimentalists and
working with the
non-Hermitian Hamiltonians
with real spectra
have two
options,
viz., the closed-system option, and the open-system
option.
In the latter case one
admits that
the
corresponding Schr\"{o}dinger equation
describes an open, unstable quantum system (such an option
and physical interpretation of non-Hermiticity are
often found advocated also
by mathematicians: see, e.g.,
Refs.~\cite{Viola,Trefethen}).

In the former case
(this option is
to be chosen and defended in what follows)
the physicists insist on the
unitarity of evolution and on the
hiddenly Hermitian
nature of
Hamiltonians.
Under these assumptions
one encounters the necessity of working, in the case of
any parameter-dependent crypto-Hermitian
Hamiltonian $H_\lambda$, with a {\em multiplet\,} of
the parallel, $\lambda-$numbered
versions of the 3HS formalism of course.

This is an apparent paradox in which
our current
research
found its motivation.
A clarification of this paradox, i.e.,
one of the main aims of our present paper will be
formulated in
section \ref{5HS}.
In particular, in paragraph \ref{PRSt}
we will explain that
in the perturbed case
it is sufficient to
keep the parameter fixed, and to
deal just with the
two (viz.,
$\lambda\neq 0$ and $\lambda=0$)
separate versions of the 3HS formalism.
In the subsequent paragraph \ref{PRSto} we describe
the merger of the two 3HS pictures which yields
the ultimate ``five Hilbert space
(5HS) formulation of quantum
mechanics''.

The amended and simplified final version of our general
non-Hermitian perturbation theory
is described in section \ref{ctyri}.
It is characterized by
a reduction of the
doublet of
the two  (viz.,
$\lambda\neq 0$ and $\lambda=0$) ``friendliest''
Hilbert spaces to the single
one,
with
$\lambda=0$.
In this space (denoted by the dedicated symbol ${\cal K}$)
the set of the ``relevant'' perturbations of $H$ is finally
introduced as the most natural form of the
dynamical-input information about the
system.

A few related technical results are then
developed and discussed in section
\ref{octyri}. Emphasis is put upon
a
simplification of the scheme based on a formal
return to a single
auxiliary Hilbert space,
and upon the study of
properties of the new, ``effective'' physical metric.
A few explicit aspects of the
construction of the
perturbation-expansion corrections
are outlined, and their
recurrent nature is underlined.
In the resulting
analogue and extension of the standard
Rayleigh-Schr\"{o}dinger perturbation-expansion recipe, the number
of the relevant Hilbert spaces will be found reduced back to three
again.

In the next section number \ref{fifi}
we will discuss a few practical aspects and
consequences of the formalism,
re-emphasizing that after the innovation, multiple
phenomenological
merits provided by the conventional
perturbative model-building strategies still remain unchanged.
In particular, we add a few more comments
on the difference between the open and closed quantum systems
(cf.
subsection \ref{opclo}) or on the problem of stability
(cf.
subsection \ref{sxxx}), etc.

In the last section \ref{susu}, our present upgrade of the
perturbation expansion techniques will be summarized.


\section{Non-Hermitian Hamiltonians and unitarity in disguise
\label{prepu}}

\subsection{Realistic Hamiltonians in non-Hermitian representation}

Since the Dyson's study of ferromagnetism \cite{Dyson},
since the
\v{C}\'{\i}\v{z}ek's coupled-cluster calculations in
quantum chemistry \cite{Cizek},
and since the successful
variational
interacting-boson-model
evaluations of energy levels in heavy
nuclei \cite{Jensen}
it has widely been accepted that at least
some of the most critical technical obstacles emerging in
multiple nontrivial
calculations can be circumvented when one weakens, in suitable
manner, the Hermiticity assumption.
A way towards such an innovation has been found to
lie in a transfer of
Eqs.~(\ref{tse}) and (\ref{selfadjo}) from the (by assumption,
overcomplicated) ``inaccessible'' textbook Hilbert space ${\cal L}$
to its unitarily non-equivalent
Hilbert-space alternative (say, ${\cal K}$).
For this purpose Dyson recommended to
pick up a suitable invertible operator $\Omega$
(carrying, typically, information about correlations)
leading to a ``prediagonalization'' of a given
realistic but overcomplicated
Hamiltonian $\mathfrak{h}=\mathfrak{h}^\dagger$,
 \be
 \mathfrak{h}=\Omega\,H\,\Omega^{-1}\,,
 \ \ \ \ \
 \Omega^\dagger\Omega=\Theta \neq I
 \,.
 \label{hait}
 \ee
In the formally
equivalent
but, presumably, simplified
time-dependent Schr\"{o}dinger equation
 \be
 {\rm i}\frac{d}{dt} |\psi(t)\kt \,
 = \,H\,|\psi(t)\kt\,,
 \ \ \ \ \ |\psi(t)\kt \, \in \, {\cal K}\,,
 \ \ \ \ H \neq H^\dagger\,,
 \label{tsetse}
 \ee
as well as in its time-independent complement
 \be
 H\,|\psi^{(n)}\kt =
 E^{(n)}\,|\psi^{(n)}\kt\,,
  \ \ \ \
 |\psi^{(n)}\kt \ \in \ {\cal K}\,,
 \ \ \ \ \ \ \
 n = 0, 1, \ldots\,
 \label{selfadjojo}
 \ee
such a transformation (also known, in numerical mathematics, as
``preconditioning'') can only be considered meaningful if it leads
to a truly perceivable simplification of the diagonalization.

\subsection{Unperturbed three Hilbert space flowchart}

The
new, isospectral form $H$ of the Hamiltonian must prove
user-friendlier. In the most successful implementations
the optimal, Dyson-recommended non-unitary
choice of mapping $\Omega$
is essential.
Naturally, the methodical gain provided by prediagonalization
$\mathfrak{h}\to H\,$ is accompanied by the loss
represented by the emergence of the non-Hermiticity of $H$
since the self-adjointness of $
\mathfrak{h}=\mathfrak{h}^\dagger$ in ${\cal L}$ gets
replaced by the more complicated
quasi-Hermiticity relation in ${\cal K}$,
 \be
 H^\dagger \Theta=\Theta\,H\,.
 \label{dudu}
 \ee
For this reason, operator
$\Theta=\Omega^\dagger\Omega$
is usually called the physical Hilbert-space metric~\cite{Geyer}.

In its abstract form the trade-off has fiercely been
criticized by Dieudonn\'{e} \cite{Dieudonne} (cf. also the most
recent warnings in \cite{Viola,ATbook}). Only the acceptance of a
few additional restrictive technical assumptions made the
recipe mathematically
satisfactory and consistent (for example, the strongly
counterintuitive boundedness of {\em all\,} of the operators of
observables was required in \cite{Geyer}).
Strictly speaking, some of the latter technical constraints may be
given a natural form when one re-qualifies the
Dyson's operator $\Omega$ as a map which connects the space ${\cal
L}$ of textbooks with another, auxiliary Hilbert space ${\cal K}$.
This proves useful, e.g.,  in the so called interacting boson models
of nuclei where the
antisymmetry of wave functions imposed
in the fermionic Fock space ${\cal
L}$
is being replaced by
the boson-space statistics in ${\cal K}$
\cite{Jensen}.
In a way reflected by
the notation this enables us to write
 \be
 |\psi\pkt = \Omega\,|\psi\kt\,,
 \ \ \ \ |\psi\pkt\ \in \ {\cal L}\,,
 \ \ \ \  |\psi\kt\ \in \ {\cal K}\,.
 \ee
Such a notation convention simplifies also another necessary
change of the Hilbert space, viz., ${\cal K} \ \to \ {\cal H}$.
This move
which completes the 3HS scheme
is easier because one merely amends the inner product,
 \be
 \br \psi_a|\psi_b\kt_{({\cal H})}=
 \br \psi_a|\Theta|\psi_b\kt_{({\cal K})}\,.
 \ee
The key purpose is that the two Hilbert spaces ${\cal L}$ and ${\cal
H}$ now become, from the point of view of making experimental
predictions, equivalent (cf., e.g., reviews \cite{Geyer,ali,MZbook}
for more details).


The Dyson-inspired constructive 3HS recipe proceeds
from a
realistic self-adjoint
$\mathfrak{h}$ and from a preselected $\Omega$  to
a friendlier $H$. Among the practical shortcomings of such a
``direct'' recipe one must underline that in  ${\cal K}$ the
upper-case
Hamiltonians with property (\ref{dudu})
are non-Hermitian. This
means that the numerical
diagonalization algorithms may be less efficient
\cite{Wilkinson}. Still, many practical applications of the
Dyson-inspired 3HS recipe \cite{Jensen} were successful while having
the following three-step structure:
 \be
  \ba
  \\
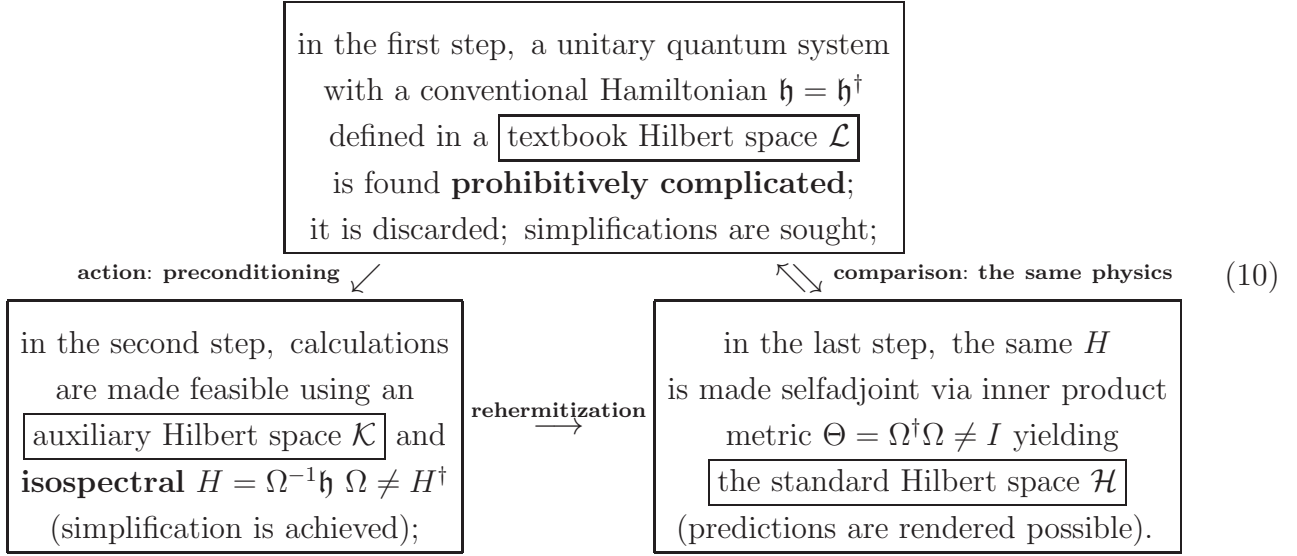
 \begin{array}{|c|}
  \hline
 \vspace{-0.35cm}\\
   {\rm in\ the\ first\ step,\ a\ unitary\ quantum\ system}\\
 {\rm with\ a\ conventional\ }
{\rm   Hamiltonian}\
 \mathfrak{h}=\mathfrak{h}^\dagger
  \\
  {\rm defined\ in\ a \ \fbox{\rm {textbook Hilbert space  ${\cal L}$}}}\\
 {\rm is\ found}\
  {\bf  prohibitively\ complicated;}\\
  {\rm it\ is\  discarded;\ simplifications \ are\ \rm sought;}\\
 \hline
 \ea
 \\
 \ \ \ \ \ \
\stackrel{{ \bf action:\ preconditioning\ } }{}
  \swarrow\ \  \  \ \ \ \ \ \
 \ \ \ \ \ \ \ \
 \ \ \ \ \ \ \ \
 \ \ \ \ \ \ \ \
 \ \ \ \ \  \nwarrow\!\!\!\searrow\
 \stackrel{\bf comparison:\  the\ same\ physics}{}
 \\
    \begin{array}{|c|}
 \hline
 \vspace{-0.35cm}\\
   {\rm in\ the\ second\ step, \ calculations}\\
 {\rm are\  made\ feasible\ using\ an } \\
   \fbox{\rm {auxiliary Hilbert space  ${\cal K}$}}{ \rm \ and}\\
     {\bf   isospectral\  } H=\Omega^{-1}\mathfrak{h}\
     \Omega \neq H^\dagger\\
   {\rm   (simplification \ is\ achieved);}
 \\
  \hline
 \ea
 \stackrel{ {\bf  rehermitization}  }{ \longrightarrow }
 \begin{array}{|c|}
 \hline
 \vspace{-0.35cm}\\
   {\rm in\ the\ last\ step,\ the \ same\ } H \\
  {\rm {is\ made\   selfadjoint\ {\rm via \ inner\ product}}}
 \\
  {\rm  \ metric  } \ \Theta=\Omega^\dagger\Omega \neq I\ {\rm yielding} \\
  {\rm  \fbox{\rm {the standard Hilbert space  ${\cal H}$}}}\\
   {\rm (predictions\ are\ rendered\ possible).\,}
 \\
 \hline
 \ea\\
 \\
 \ea
 \label{trihs}
 \ee
Typically, as we already emphasized,
the fermionic, Pauli-principle-controlled Fock-space
choice of ${\cal L}$ and the simpler, bosonic-space choices of
${\cal K}$ and ${\cal H}$ together with the Hermiticity of
 \be
 H  = \Theta^{-1}H^\dagger\Theta := H^\ddagger
 \ee
in ${\cal H}$ opened the way towards the
unitary equivalence between the two alternative physical Hilbert
spaces ${\cal H}$ and ${\cal L}$.

\section{Inverted unperturbed flowchart\label{realth}}


\subsection{The necessity of the reconstruction of
metric\label{yamb}}

Once we decide to start, directly, from a suitable $H$ in ${\cal
K}$, we have to follow the recipes by Scholtz et al \cite{Geyer}, or
by Bender with coauthors \cite{Carl}. The necessary process of the
Hermitization ${\cal K} \to {\cal H}$ has to involve not only $H$
but also, possibly, the knowledge of a quasi-parity ${\cal Q}$
\cite{quasi}, of a charge ${\cal C}$ \cite{Carl}, or of any other
``compatible'' \cite{arabky} candidate $\Lambda$ for an observable
\cite{Geyer}. The ultimate aim of the calculations then becomes the
reconstruction of metric $\Theta(H)$  \cite{lotor}. Naturally, the
formalism still can provide all of the relevant experimental
predictions as well as a consistent picture of physical reality, in
principle at least. In practice, the evaluation of predictions will
be guided by the following inverted triangular diagram,
 \be
  \ba
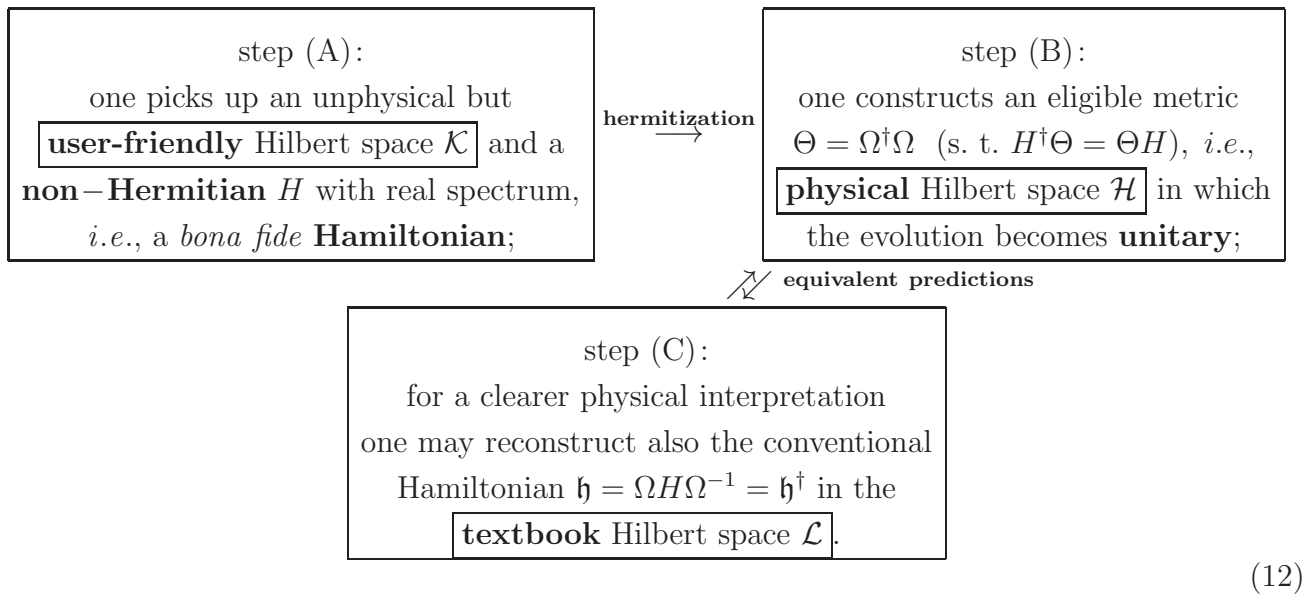
    \begin{array}{|c|}
 \hline
 \vspace{-0.35cm}\\
   {\rm step \ (A)\!:}\\ {\rm one\ picks\ up\ an\ {unphysical} \ but}\\
  {\rm  \fbox{\rm {{\bf user-friendly} Hilbert space  ${\cal K}$}}\ and\ a}\\
     {\rm \bf non\!-\!Hermitian\  } H\
     {\rm with\ real\ spectrum,}\\
    {\rm {\it i.e.,}\  a\ \bf {\it bona\ fide} \  Hamiltonian};
 \\
  \hline
 \ea
 \stackrel{ {\bf  hermitization}  }{ \longrightarrow }
 \begin{array}{|c|}
 \hline
 \vspace{-0.35cm}\\
   {\rm step\ (B)\!:}\\ {\rm  one \ constructs\ an\ eligible } \
  {\rm metric\ } \\
  \Theta=\Omega^\dagger\Omega \ \ ({\rm s.\ t. }\
   H^\dagger \Theta=\Theta
  H),\ i.e.,\\
  {\rm \fbox{\rm {{\bf physical} Hilbert space  ${\cal H}$}}}\
   {\rm in \ which }\\
   {\rm the\ evolution\ becomes\ \bf unitary;} \\
   %
 \hline
 \ea\\
 \ \ \ \ \ \ \ \
 \ \ \ \ \ \ \ \
 \ \ \ \ \ \ \ \
 \ \ \ \ \ \ \ \
 \ \ \ \ \ \ \ \
 \ \ \ \ \  \nearrow\!\!\!\swarrow\
 \stackrel{\bf  equivalent\ predictions}{}
 \\
 \begin{array}{|c|}
  \hline
 \vspace{-0.35cm}\\
   {\rm step\ (C)}\!:\\
  {\rm {for\ a}\ clearer \ {\rm physical}\ interpretation} \\
  {\rm  one \ may\ reconstruct\ also\ the\  conventional }\\
 {\rm  Hamiltonian }
 \
 \mathfrak{h}
 =\Omega H\Omega^{-1}=
 \mathfrak{h}^\dagger\ {\rm in\ the}\\
  {\rm  \fbox{\rm {{\bf textbook} Hilbert space  ${\cal L}$}}.}\\
 \hline
 \ea
 \\
 \ea
 \label{humandiag}
 \ee
This characterizes the theoretical construction flowchart used, in
particular, when one speaks about the ``non-Hermitian but ${\cal
PT}-$symmetric'' Hamiltonians $H$ \cite{Carl}, or about the
${\cal P}-$pseudo-Hermitian Hamiltonians $H$ \cite{ali}, etc.
An explicit two-by-two matrix illustration of
the backward assignment $H \
\to \ \Theta(H)$ will be provided in the following subsection.

\subsection{The ambiguity of the reconstruction of
metric\label{ambb}}

Operator $\Theta$ is, in general, non-universal and
Hamiltonian-dependent. Its use is easily shown to convert the
conventional Hermiticity $\mathfrak{h}=\mathfrak{h}^\dagger$ of
Hamiltonian (\ref{hait}) in ${\cal L}$ into the formally equivalent
quasi-Hermiticity property
 \be
 H = \Theta^{-1} H^\dagger\Theta\,
 \label{cryher}
 \ee
of its non-Hermitian preconditioned
isospectral partner in ${\cal K}$.

A suitable elementary but still sufficiently realistic $2N$ by
$2N$-matrix example of the formalism may be found, e.g., in
Ref.~\cite{Semoradova}. The physical contents of the theory has been
illustrated there, in the context of relativistic quantum mechanics,
using the discrete version of the free Klein-Gordon Hamiltonian
$H^{(KG)}$. Even the first nontrivial special case of this model,
with $N=1$, offers a nontrivial but exactly solvable two-level
quantum system characterized by the real two-by-two-matrix
Hamiltonian
 \be
 H=H^{(KG)}(\tau)=
 \left[ \begin {array}{cc} 0&\exp 2 \tau\\
 \noalign{\medskip}1&0\end {array} \right]
 \neq H^\dagger\,.
 \label{primo}
 \ee
The dynamics is controlled here by a single real parameter $\tau \in
(-\infty,\infty)$. The toy-model Hamiltonian (\ref{primo}) is ${\cal
PT}-$symmetric, with antilinear ${\cal T}$ representing Hermitian
conjugation, and with the Pauli-matrix parity ${\cal P}=\sigma_x$.
In spite of the manifest non-Hermiticity of $H^{(KG)}(\tau)$ at
$\tau \neq 0$, the bound-state energies are real and non-degenerate,
$E= E_\pm=\pm \exp \tau$. The model is suitable for illustrative
purposes because the list of its properties survives a transition to
the discrete free Klein-Gordon Hamiltonians with arbitrary matrix
dimensions $2N$ (cf. \cite{jaKG}).

Every element $H^{(KG)}$ of the latter family is non-Hermitian but
it still has been shown to admit the conventional probabilistic
interpretation. This goal is achieved in the third Hilbert space
${\cal H}$ in which, by construction, the eigenstates of the
upper-case Hamiltonian $H$ really do acquire their conventional
probabilistic physical interpretation of wave functions. In
comparison, the role of the friendlier
Hilbert space ${\cal K}$ (with non-amended inner product) remains
just that of a mathematical playground.

The construction of the correct physical inner product (i.e., of the
non-equivalent space ${\cal H}$) is a key to the phenomenological
acceptability of the model. For illustrative Eq.~(\ref{primo}) the
space ${\cal K}={\cal K}^{(KG)}$ is the real two-dimensional vector
space. In such a model it is easy to show that the Hamiltonian
$H^{(KG)}$ will satisfy the ``hidden'', ${\cal H}-$space-Hermiticity
relation (\ref{cryher}) if and only if the following, entirely
general amended form
 \be
 \Theta=\Theta^{(KG)}(\tau,\beta)=
 \left[ \begin {array}{cc} \exp (-\tau)&\beta\\
 \noalign{\medskip}\beta&\exp \tau\end {array} \right]
 =\Theta^\dagger\,
 \label{secuf}
 \ee
of the metric is used upgrading ${\cal K}^{(KG)}$ into our ultimate
physical Hilbert space ${\cal H}^{(KG)}$.

Matrix (\ref{secuf}) contains also another, independent real
variable $\beta$ which must keep it positive definite \cite{Geyer},
i.e., which must be such that $|\beta|<1$. Its variability reflects
the generic ambiguity of the assignment $H^{(KG)} \to \Theta^{(KG)}
\to {\cal H}^{(KG)}$. The free variability of such a ``redundant''
parameter must be kept in mind as an inseparable, characteristic
part of the 3HS formulation of quantum theory \cite{SIGMAdva}.

According to review paper \cite{Geyer} there exists the
possibility of a systematic, physics-based removal of the ambiguity
mediated by the completion of the set of the so called irreducible
observables (say, $\Lambda_j$, $j=1,2, \ldots, j_{\max}$).
In our most elementary real-matrix model (\ref{primo})
leading to the single-parameter
ambiguity in metric
(\ref{secuf}), for example, one would need just $j_{\max}=1$,
i.e., just one
``missing''  real-matrix candidate for observable
 \be
 \Lambda=\left[ \begin {array}{cc} a&b\\
 c&d\end {array} \right]\,.
  \label{desecuf}
 \ee
As long as its observability property
$\Lambda^\dagger \Theta^{(KG)}(\tau,\beta)
=\Theta^{(KG)}(\tau,\beta) \Lambda$
forms the set of four linear algebraic equations,
we have to distinguish between
the solutions with $d=a$ and $d \neq a$.
As long as the former case (of which the Hamiltonian
itself is the special case with $a=0$)
does not determine the value of $\beta$, all of
the eligible candidates $\Lambda$
belong to the latter case. For all of them
our set of four equations has the solution
 \be
 \beta=\beta(a,b,c,d,\tau)=
 \frac{c\,\exp (\tau)-b\,\exp (-\tau)}{d-a}\,.
 \label{resusu}
 \ee
Thus, we proved the following elementary but
instructive result.
\begin{lemma}
Once we guarantee that
the four eligible real parameters in the
``additional dynamical input'' observable (\ref{desecuf})
are such that $\ |\beta(a,b,c,d,\tau)|<1$ in Eq.~(\ref{resusu}),
the corresponding
correct physical Hilbert-space metric (\ref{secuf})
becomes unique.
\end{lemma}

\section{Five-Hilbert-space formulation of quantum mechanics
\label{5HS}}


The three-Hilbert-space (3HS) formulation of quantum mechanics
as reviewed in the preceding section
offers an explanation why,
under certain general mathematical assumptions, any given
parameter-independent
quasi-Hermitian Hamiltonian $H$
with real spectrum
can play the role of
the generator of quantum evolution
which is unitary.
In essence, the explanation
is based on the
clarification
that only the two Hilbert spaces in question (viz., ${\cal L}$ and
${\cal H}$) are physical (i.e., unitarily equivalent) while the third
Hilbert space ${\cal K}$
(in which calculations are being performed)
yields unphysical
values of inner products.

\subsection{Parameter-dependent 3HS theory \label{PRSt}}

In the overall 3HS framework the
correct physical interpretation of a quantum system in question
becomes much more complicated when the upper-case Hamiltonians
are allowed to vary with a parameter \cite{arabky}.
Even the inspection of our
most elementary $\tau-$dependent toy model (\ref{primo})
reveals that the change of the parameter
may modify the
physical Hilbert-space metric.
Moreover,
even at a fixed $\tau$
the specification of the metric remains ambiguous
-- notice that its
form (\ref{secuf}) varies with another free parameter $\beta$.
After a change of $\tau$, our choice of $\beta$
may be different, $\beta=\beta(\tau)$.

For an arbitrary
crypto-Hermitian 3HS
Hamiltonian (or, if needed, for the
operators of arbitrary other
observable quantities)
both of the latter remarks are truly essential.
In particular, they are essential for an appropriate
and consistent introduction of the concept of perturbations
of the Hamiltonians in Schr\"{o}dinger picture \cite{Ruzicka}.
Indeed, any deviation
from the ``unperturbed'' scenario, even with a ``very small''
parameter $\lambda$ in
 \be
 H \ \to \ H_\lambda=H+{\cal O}(\lambda) \neq H^\dagger_\lambda\,
 \label{conconve}
 \ee
makes it necessary to reopen the question of the
applicability of the conventional Rayleigh-Schr\"{o}dinger
power series in $\lambda$.

Our present main task is to demonstrate
why the perturbed forms (\ref{conconve})
of the 3HS-related non-Hermitian upper-case
Hamiltonians really require a deep change of
the conventional perturbation-approximation approaches.
We will show under which conditions these deeply anomalous
models can still be made a part of a
consistent quantum theory of unitary evolution.

Let us start by recalling that whenever one alters the Hamiltonian
(say, by making it, via Eq.(\ref{conconve}), parameter-dependent),
the parameter-dependence becomes inherited
also by the related Hilbert-space triplet.
One of the most important consequences of the generalization
 \be
 \{ {\cal L}, {\cal K}, {\cal H}
 \}
 \ \to \
 \{ {\cal L}_\lambda, {\cal K}_\lambda, {\cal H}_\lambda
 \}
 \label{parade}
 \ee
is that
the
standard concept of the smallness of
perturbations
retains its intuitively acceptable physical meaning only
in the ``direct'' flowchart scheme of Eq.~(\ref{trihs}), i.e.,
only in the textbook Hilbert space ${\cal L}$.
The principle of correspondence enables us to
quantify there what is a ``small'' and ``large''
perturbation term in (\ref{conve}).
No such a guiding principle survives the transition to the
upper-case perturbed Hamiltonians.

The key point is that in the crypto-Hermitian picture
even the very physical
interpretation of the observable quantities becomes
parameter-dependent (cf. Ref.~\cite{Geyer} or Eq.~(\ref{parade})).
We know (see paragraph \ref{ambb}
or our recent comment \cite{lotor}) that
in the
``inverted'' 3HS
flowchart (\ref{conconve})
one encounters also
the
emergence of new
free parameter(s) in metric. In other words,
the ambiguity is enhanced.
This implies the necessity of
selection of a
method of the
removal of the formal ambiguity
of the metric. The reconstruction of physics
(i.e., in essence, a return to
${\cal L}$) has simply the
meaning of a return to the principle of correspondence.

\subsection{A merger of two Hilbert-space triplets requiring
${\cal L}_\lambda={\cal L}$\label{PRSto}}

The Dieudonn\'{e}'s analysis of the possible formulation of quantum
mechanics covering the general unbounded quasi-Hermitian
Hamiltonians \cite{Dieudonne} did not find, due to its discouraging
sceptical tone, too many applications in physics. The idea was only
accepted and became popular after Bender with coauthors
\cite{BB,Carl} proposed the restriction of the class of
non-Hermitian Hamiltonians with real spectra
to a ``tractable'' subset
exhibiting
${\cal PT}$ symmetry
(i.e., antilinear parity-times-time-reversal
symmetry).
With multiple technical details covered by
reviews (e.g., \cite{Carl,ali}) and books (e.g.,
\cite{book,Carlbook}), the key message was that the
${\cal PT}$ symmetry assumption
leads to a thorough simplification of the
technicalities.

In the more general setting, without a direct reference to
${\cal PT}$ symmetry, a thorough review of the eligible
techniques of construction of the metric may be found in \cite{ali}.
One has to admit that even when a given Hamiltonian $H$ with real
spectrum is successfully assigned a suitable ``Hermitizing'' metric
$\Theta=\Theta(H)$, the result is heavily ambiguous
\cite{lotor,SIGMAdva} plus, in any case, feasible just for the
sufficiently elementary models (cf., e.g., \cite{117,david}).

One is
forced to expect that the difficulties will further grow if the
initial Hamiltonian $H$ gets perturbed. Even the simplifications
resulting from the assumptions of ${\cal PT}$ symmetry may fail to
apply when one tries to study the role and impact of perturbations.
For example, a
constructive disproof of applicability of the perturbation
expansions of bound states using the conventional
Rayleigh-Schr\"{o}dinger power-series ansatzs using the
{\it bona fide\,} ``small''
modification (\ref{conconve}) of the Hamiltonian
near an exceptional-point singularity was given in
\cite{admissible}.
In our present paper, therefore, we will assume that the
unperturbed Hamiltonian $H$
is safely diagonalizable and lying
sufficiently far from the influence of the possible
exceptional-point-related pathologies.

Once we pick up some non-vanishing value of $\lambda \neq 0$
in the perturbed Hamiltonian,
the most natural
(and, in fact, the only properly interpreted)
merger of this case with its unperturbed, exactly solvable
$\lambda=0$
predecessor may be realized within the
textbook Hilbert space ${\cal L}$.
In this sense we may identify ${\cal L}_\lambda$ with ${\cal L}$.
At this instant
our theoretical analysis of perturbations
may start by a
return to the unperturbed Hamiltonian $H$ (defined in ${\cal K}$)
and to the related (and, presumably, also explicitly known)
physical metric $\Theta$ in Hilbert space ${\cal H}$.
Subsequently, an additional,
methodically motivated assumption of our knowledge of the
Dyson-map factors $\Omega$ of the metric would
open the way towards the reconstruction of the
unperturbed physics in ${\cal L}$.
In the latter space we may
introduce the perturbations, purely formally,
by Eq.~(\ref{conve}).
This completes the
merger of the unperturbed ``inverted''
3HS diagram  (\ref{humandiag})
with its Dyson-like ``direct'' 3HS
perturbation-containing
extension  (\ref{trihs}).
The following, most natural five-Hilbert-space (5HS) perturbation-theory
flowchart
is obtained as a result,
 \be
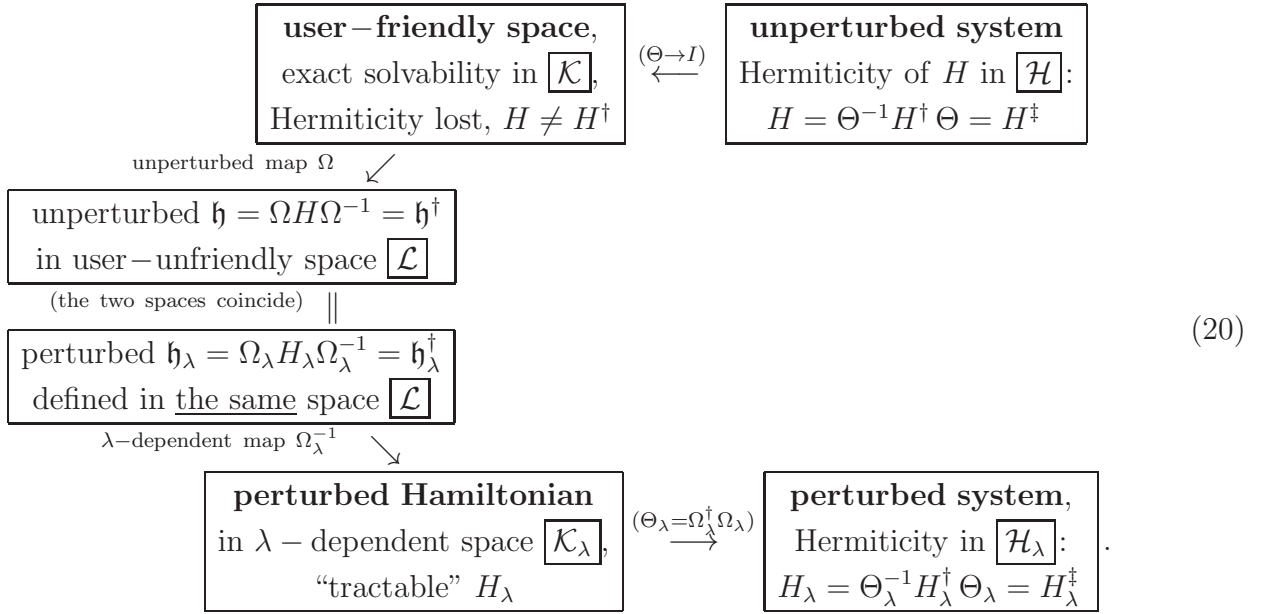
  \begin{array}{l}
  \ \ \ \ \ \ \ \ \ \
  \ \  \  \ \ \ \ \ \
  \ \  \  \ \
 \begin{array}{|c|}
 \hline
 {\rm  \bf user\!-\!friendly \  space,}\\
 {\rm  exact\ solvability \ \rm in}\
 \fbox{$\cal K$},\\
 {\rm Hermiticity\ lost,}\
  H \neq H^\dagger
 \\
  \hline
 \ea
 \stackrel{}
 { \stackrel{\,(\Theta \to I)\ {\bf }}{\longleftarrow } }
 \begin{array}{|c|}
 \hline
   {\bf unperturbed\ system} \\
 {\rm  {Hermiticity}\ of}\ H \ {\rm    in}\
  \fbox{$\cal H$}\!:
 \\
    H=\Theta^{-1}
 H^\dagger\,
 \Theta =H^\ddagger\,
 \\
 \hline
 \ea
\\ \ \ \
  \ \  \  \ \ \ \ \ \
 %
 \stackrel{{\rm unperturbed \  map}\ \Omega}
 \ \ \ \
  \swarrow\ \  \  \ \ \ \ \ \
 \ \ \ \ \ \ \ \
 \ \ \ \ \ \ \ \
  \ \  \  \ \ \ \ \ \
  \ \  \  \ \ \ \ \ \
 \ \ \ \ \  
 \\
    \begin{array}{|c|}
   \hline
 {\rm \ unperturbed}\
  \mathfrak{h}
 =\Omega H\Omega^{-1}=\mathfrak{h}^\dagger
  \\
   {\rm in\ user\!-\!unfriendly\ }
  {\rm  space\ }\fbox{$\cal L$}
 \\
  \hline
    \end{array}\\
    \ \ \ \
 \stackrel{{\rm (the\ two\ spaces\ coincide) }}
 \ \ \ \parallel
    \\
    \begin{array}{|c|}
 \hline
  {\rm perturbed} \
 \mathfrak{h}_{{\lambda}}=\Omega_{{\lambda}} H_{{\lambda}}\Omega^{-1}_{{\lambda}}
 =\mathfrak{h}_{{\lambda}}^\dagger
  \\
   {\rm defined\ in\ } \underline{\rm  the \ same}\
  {\rm  space\ }\fbox{$\cal L$}
 \\
 \hline
 \ea
  \ \  \  \ \ \ \ \ \
  \ \  \  \ \ \ \ \ \
  \ \  \  \ \ \ \ \ \
 \\
  \ \  \  \ \ \ \ \ \
 \stackrel{ \lambda-{\rm dependent\ map }\ \Omega_{{\lambda}}^{-1}}
 \ \ \ \
  \searrow\ \  \  \ \ \ \ \ \
 \ \ \ \ \ \ \ \
  \ \  \  \ \ \ \ \ \
 \ \ \ \ \ \ \ \
  \ \  \  \ \ \ \ \ \
 \ \ \ \ \
 \stackrel{{\bf }}{}\\
  \ \  \  \ \ \ \ \ \
   \ \ \ \ \ \
  \ \  \  \
  \begin{array}{|c|}
 \hline
 {\rm  {\bf perturbed\ Hamiltonian} }\\
 {\rm in}\
  {\lambda}-{\rm dependent\ space}\
 \fbox{${\cal K}_{{\lambda}}$},\\
 {\rm   ``tractable" }\
 H_{{\lambda}}
  \\
%
%
  \hline
 \ea
 \stackrel
 { \stackrel{{\bf  }}{(\Theta_{{\lambda}}=\Omega^\dagger_{{\lambda}}\Omega_{{\lambda}})} }
 { \longrightarrow }
 \begin{array}{|c|}
 \hline
 {\rm  \bf  perturbed\ system,}
 \\
  {\rm  {Hermiticity}}\ {\rm    in}\ \fbox{${\cal H}_{{\lambda}}$}\!:\\
 %
    H_{{\lambda}}=\Theta^{-1}_{{\lambda}}
 H^\dagger_{{\lambda}}\,
 \Theta_{{\lambda}} =H^\ddagger_{{\lambda}}\,\\
 \hline
 \ea
 \,.
 \\
\\
 \ea
 \\
 \\
 \label{findiag}
 \ee
Obviously, an immediate consequence of
our decision of the identification of the two textbook
Hilbert spaces
${\cal L}$ and ${\cal L}_\lambda$
and of starting from the solvable unperturbed Hamiltonian $H$
is that in applications,
we have to have access to
the dynamical input information
in its ``inaccessible'', ${\cal L}-$space-based form
(\ref{conve}) of the
lower-case operator $\mathfrak{v}$.
In applications we would encounter the necesssity of
redirecting the 5HS flowchart of Eq.~(\ref{findiag}).

\section{Inverted 5HS flowchart and its simplifications\label{ctyri}}

In diagram (\ref{findiag}) it is natural to require
the full knowledge of the
unperturbed system, i.e., of the upper case Hamiltonian $H$ acting
in ${\cal K}$,
of its metric $\Theta$ and factors $\Omega$,
and of the lower-case Hamiltonian $\mathfrak{h}$
acting in ${\cal L}$.
As already indicated, the rest of the flowchart may be then read as
an analogue of the ``direct'' pattern of Eq.~(\ref{trihs}).
Unfortunately, this only allows one to start from a
lower-case perturbation
term ${{\lambda}}\,\mathfrak{v}_{{\lambda}}$ in ${\cal L}$
which is, due to the overall
Dyson-inspired philosophy, prohibitively complicated.
The simplifications, it any, can only be based on
an intuitive insight helping us to pick up an
appropriate
{\em ad hoc\,} Dyson map $\Omega_\lambda$ leading to
the
expected thorough simplification of the new Hamiltonian
$H_\lambda$.

In general, one can hardly rely upon the similar intuitive guidance,
and one has to start perturbation expansions
directly from the hypothetical dynamical input knowledge of
the upper-case perturbation.
In other words, at least some of the arrows in our 5HS flowchart
(\ref{findiag}) have to be inverted.

\subsection{Formal elimination of ${\cal K}_\lambda$\label{tyri}}

In the 5HS recipe the first, most visible challenge
lies in the manifest $\lambda-$dependence of the working
Hilbert space ${\cal K}_\lambda$. A remedy consists in
a transition from the $\lambda-$dependent space
${\cal K}_\lambda$ to its $\lambda-$independent version ${\cal K}$.
The purpose can be served by an abstract operator ${
J}_{{\lambda}}$ which would link the two spaces,
 \be
  |\psi\kt_{{\lambda}} \in {\cal K}_{{\lambda}}\,,
 \ \ \ \
 |\psi\kt={ J}_{{\lambda}}|\psi\kt_{{\lambda}} \in {\cal K}\,.
 \ee
The perturbed Hamiltonian is defined in ${\cal K}_{{\lambda}}$ and
reads
 \be
 H_{{\lambda}}=
 \Omega_{{\lambda}}^{-1}\mathfrak{h}_{{\lambda}}\Omega_{{\lambda}}
 \,.
 \label{23}
 \ee
Only now we should formally introduce
a ``tractable'' upper-case operator $W$ representing the
input dynamical information about the perturbation,
 \be
  \left ({ J}^{\dagger}_{{\lambda}}\right )^{-1}
 H_{{\lambda}}\,{ J}_{{\lambda}}^{-1}
 =
 H+\lambda\,W_\lambda\,.
 \label{21}
 \ee
Such a form of perturbation is defined in ${\cal K}$. It should be
treated as our initial information on dynamics. The 5HS flowchart
(\ref{findiag}) must be, in this sense, not only inverted
but also re-read -- as a
perturbation-theory analogue of the 3HS reconstruction recipe
prescribed by diagram~(\ref{humandiag}).

Fully in the spirit of the specific Bender's interpretation of
the theory \cite{Carl}, also
one of our present most fundamental methodical requirements is that the
quantum system in question is not an open system and that its
evolution is unitary. In ${\cal L}$ the perturbation term must be
self-adjoint, $\mathfrak{v}_\lambda= \mathfrak{v}^\dagger_\lambda$.
Otherwise, we could not accept any given candidate for the
Hamiltonian, irrespectively of its particular representation, as an
operator of an experimentally realizable observable. The emergence
of {\em any\,} non-Hermiticity in $\mathfrak{v}_\lambda$ would
indicate that our system ceased to be isolated from its environment.
Being exposed to an external force, weak or strong, all of its
measurable features could suffer a decay or growth.

Naturally, the unitarity restriction and the Hermiticity assumption
are widely accepted when one works in the conventional framework of
textbook space ${\cal L}$. These requirements must be translated
into the language of the sophisticated, $\lambda-$dependent physical
Hilbert space ${\cal H}_\lambda$. This being said, the calculations
have to be performed in, and only in,
Hilbert space  ${\cal K}$. Only the standard probabilistic
interpretation of the results must be pulled down and
performed, via metric
$\Theta_\lambda$, in ${\cal H}_\lambda$.


\subsection{The effective physical metric $\,T_\lambda\,$ in  $\,{\cal K}-$space}



The 3HS theories may be characterized, for unperturbed as well as
perturbed systems, by the separation of descriptions of mathematics
(in ${\cal K}_{{\lambda}}$) and of physics (in ${\cal
H}_{{\lambda}}$). Such a split of the two traditional roles of the
single textbook Hilbert space ${\cal L}$ may simplify the picture. A
deeper analysis of the change may be based on the factorization of
the perturbed Dyson's map into its unperturbed part and a correction
which is, presumably, small,
 \be
 \Omega_{{\lambda}}=\Omega\,(1+{{\lambda}}\,\Delta_{{\lambda}}){
 J}_{{\lambda}}\,.
 \label{25a}
 \ee
In the next step we find a map between ${\cal K}_{{\lambda}}$ and
${\cal K}$ and, subsequently, we factorize the metric accordingly,
 \be
 \Theta_{{\lambda}}
 ={ J}^\dagger_{{\lambda}}\,T_{{\lambda}}\,{ J}_{{\lambda}}\,,
 \ \ \ \ \
 T_{{\lambda}}=(1+{{\lambda}}\,\Delta^\dagger_{{\lambda}})\,\Theta\,
 (1+{{\lambda}}\,\Delta_{{\lambda}})\,.
 \label{26}
 \ee
This enables us to parallel the reconstruction step (B) of the
general 3HS recipe (\ref{humandiag}). As long as we have
$T_{{\lambda}} \neq \Theta_{{\lambda}}$ and ${\cal K}_\lambda \neq
{\cal K}$, and as long as the metric depends on ${{\lambda}} \neq
0$, the reconstruction does not yield the original physical Hilbert
space ${\cal H}={\cal H}_{{\lambda}}$ of diagram (\ref{findiag}) but
rather another physical Hilbert space. This space would be
an additional component of the theory. It deserves to be denoted by
a separate dedicated symbol, say, ${\cal M}_{{\lambda}}$. Thus, we
arrive at the following operational interpretation of the scheme.
\begin{lemma}$\!\!\!.$
 \label{te27}
After a transfer of the theory from ${\cal K}_{{\lambda}}$ to ${\cal K}$
via Eqs.~(\ref{23}) and (\ref{21}),
the crypto-Hermiticity constraint
(\ref{cryher}) imposed upon the Hamiltonian may be given the form
 \be
 (H^\dagger + {{\lambda}}\,W^\dagger)\,T_{{\lambda}}
 =T_{{\lambda}}\,(H + {{\lambda}}\,W)\,.
 \label{27}
 \ee
\end{lemma}
Given the perturbation ${{\lambda}}\,W$, this equation specifies the
(empty or non-empty) class of admissible (and, in general,
non-unique) inner-product metrics $T_{{\lambda}}$ converting our
working Hilbert space ${\cal K}$ into an innovated
{physical} Hilbert space ${\cal M}_{{\lambda}}$.
Diagrammatically, we have
  \be
  \vspace{0.1cm}
  \ba
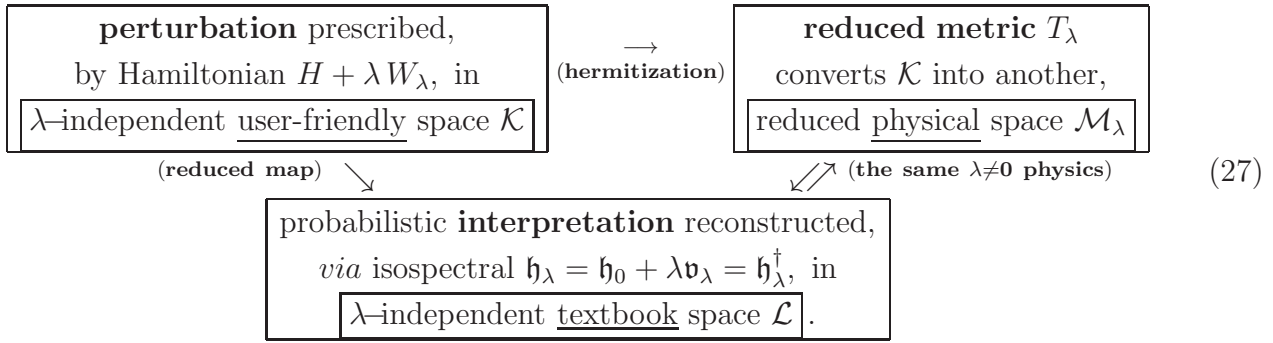
  \begin{array}{|c|}
 \hline
 {\rm   \bf perturbation\ \rm  prescribed, }\\
 {\rm by\ Hamiltonian}\
 H+{{\lambda}}\,W_\lambda,\
  {\rm in}
  \\
 \fbox{{\rm  $\lambda\!-\!$independent \underline{user-friendly} space ${\cal K}$}}\\
%
%
  \hline
 \ea
 \stackrel{ \longrightarrow }
 { \stackrel{{\bf (hermitization)}}{} }
 \begin{array}{|c|}
 \hline
 { \bf reduced\   metric \   }  T_{{\lambda}}  \\
 {\rm   converts\ }{\cal K}
 {\ \rm into\ another, }
 \\
 \fbox{{\rm reduced \underline{physical} space ${\cal M}_{{\lambda}}$}}\\
 \hline
 \ea
 \\
 %
 \ \
  \ \  \  \ \ \ \ \ \
 \stackrel{{\bf (reduced\ map)}}{}
 \ \
  \searrow\ \  \  \ \ \ \ \ \
 \ \ \ \ \ \ \ \
  \ \  \  \ \ \ \ \ \
  \ \  \  \ \ \ \ \ \
 \ \ \ \ \  \swarrow\!\!\!\nearrow\
 \stackrel{{\bf (the\ same\ {{\lambda}}\neq 0\ physics)}}{}
 \\
    \begin{array}{|c|}
 \hline
 {\rm  probabilistic\  \bf interpretation\ \rm reconstructed,}\\
 via\ {\rm    {isospectral}}\
 \mathfrak{h}_{{\lambda}}
 =\mathfrak{h}_0+{{\lambda}} \mathfrak{v}_\lambda
  =\mathfrak{h}_{{\lambda}}^\dagger,
  \
  {\rm in}\\
  \fbox{{\rm $\lambda\!-\!$independent \underline{textbook} space ${\cal L}$}}\,.\\
 \hline
 \ea
 \\
 \ea
  \vspace{0.1cm}
 \label{dvacetosm}
 \ee
This represents the ultimate compact pattern suitable for the
treatment of perturbations in quasi-Hermitian or ${\cal
PT}-$symmetric quantum theories. We have to conclude that in these
upgrades of the conventional formulation of quantum mechanics in
Schr\"{o}dinger picture the answer to the key question ``Which
perturbations are small?'' is much less straightforward than one
would {\it a priori\,} expect.

\section{Perturbation series in
crypto-Hermitian physics\label{octyri}}

In our preceding analysis we demonstrated that in the theories using
the non-selfadjoint representations of observables a proper
treatment of perturbations is far from trivial. We managed to
simplify the general mathematical pattern but its implementations
still pose a few open problems.

\subsection{Upper-case
perturbations $\lambda\,W(\lambda)$ in  $\,{\cal K}-$space\label{poctyri}}

One of the important tacit assumptions lying behind our preceding
considerations is that the unperturbed Hamiltonian $H$ is
well-behaved and diagonalizable. Naturally, in an overall context of
${\cal PT}-$symmetric quantum mechanics such a property of $H$ is
fragile because in general, such an operator may still exhibit a
high degree of non-Hermiticity in ${\cal K}$. In such a case a word
of warning was issued in our recent paper \cite{corridors}. We
studied there a specific family of toy models in which the
Hamiltonians of interest were ``almost equal'' to a
non-diagonalizable Jordan-block $N$ by $N$ matrix. In the paper the
model has been found ``almost non-perturbative'' simply because its
Hamiltonians $H$ were all lying very close to the Kato's \cite{Kato}
exceptional-point dynamical extreme. For this reason
even the ``unperturbed'' metric $\Theta$ happened to be (almost)
singular. This means that the model itself did not certainly satisfy
the mathematical assumptions of smoothness and applicability of
flowcharts (\ref{findiag}) or (\ref{dvacetosm}).

Another, more physics-oriented word of warning was formulated in our
brief comment \cite{Ruzicka}. We emphasized there that one of the
key distinguishing features of the {\em closed\,} quasi-Hermitian
quantum systems (i.e., e.g., of those exhibiting the spontaneously
unbroken ${\cal PT}-$symmetry \cite{Carl}) is that the corresponding
{\em real\,} manifolds of parameters may contain certain
submanifolds (mostly, the hierarchies of hypersurfaces) which are
formed by the associated Kato's exceptional points. Naturally, these
hypersurfaces determine the phase-transition boundaries of the
(often, compact \cite{maximal}) parametric domains of stability
${\cal D}$ of the quantum system in question. Again, we have to
emphasize that in any prospective analysis of stability, the
applicability of our present constructive flowcharts only remains
restricted to the ``deep interiors'' of ${\cal D}$ admitting the
standard unitary-evolution interpretation
of the closed quantum system in
question.

Under these assumptions we may expect that the pull-down process
${\cal K}_\lambda \to {\cal K}$ would not lead to any conceptual
problems. We may merely require the knowledge of an appropriate
input form (\ref{21}) of the Hamiltonian. In the context of physics
we can also make use of the non-singular one-to-one correspondence
between the textbook space ${\cal L}$ and its unitarily
non-equivalent alternative ${\cal K}_\lambda$. The
existence of the invertible $\lambda-$dependent
Dyson map (\ref{25a}) is then found equivalent to the
stability of the quantum system in question.

\begin{lemma}$\!\!\!.$
Whenever the spectrum of the
perturbed
Hamiltonian $H+{{\lambda}}\,W_\lambda$ (which is defined and
non-Hermitian in the reduced, $\lambda-$independent Hilbert space ${\cal K}$)
is not real,
the positive definite solution $T_\lambda=T^\dagger_\lambda$ of
Eq.~(\ref{25a}) does not exist.
\label{ledva}
\end{lemma}

\begin{proof}
As long as definition (\ref{26})
implies the positive definiteness of metric $T_{{\lambda}}$
in ${\cal M}_{{\lambda}}$ at any not too large $\lambda$,
relation (\ref{27}) should be re-read
as a hidden Hermiticity requirement,
i.e., as a spectral reality condition imposed upon
our operator $H+{{\lambda}}\,W_\lambda$ in ${\cal K}$.
\end{proof}
The perturbations-incorporating 5HS formalism
offers an expected and consistent picture of correspondence between
the loss of the reality of the spectrum (occurring at the Kato's
exceptional-point limit of $\lambda \to \lambda^{(EP)}$) and the
loss of the unitary equivalence between the respective user-friendly
and user-unfriendly physical Hilbert spaces ${\cal M}_\lambda$ and
${\cal L}$. In and only in this sense we may formally separate the
class of perturbations $\lambda\,W_\lambda$ into its ``sufficiently
small'' and ``inadmissibly large'' subclasses. At the same time, due
to the manifest $\lambda-$dependence of ${\cal M}_\lambda$, the
formulation of some useful criteria of such a split would be much
more difficult than in conventional quantum mechanics (a persuasive
illustrative example may be found, e.g., in \cite{corridors}).

\subsection{Correspondence principle and lower-case
perturbations $\lambda\,\mathfrak{v}(\lambda)$\label{gyri}}

%


At every not too large value of $\lambda$ the statement of Lemma
\ref{ledva} can be inverted in the sense that the reality of the
non-degenerate perturbed spectrum provides access to the
$\lambda-$dependent physical Hilbert space ${\cal M}_\lambda$. But
this would be an observation of little practical value because we
only know the unperturbed spectrum in advance. This is an obstacle
which could invalidate the routine direct analysis of stability of
given quantum systems with respect to random perturbations
\cite{Ruzicka}. At the same time, in many quantum models, the
``inaccessible'' Hilbert space ${\cal L}$ still keeps carrying the
role of a reliable interpretation background using the so called
principle of correspondence.

The applicability of such a principle to the closed quantum models
defined directly in ${\cal K}$ appears, unfortunately, strongly
limited. What remains at our disposal are only the possibilities of
an indirect treatment of the correspondence and stability issues.
Its realization could be based on a transfer of some time-proven
requirements (e.g., of the norm-boundedness of the lower-case
perturbations $\lambda\,\mathfrak{v}(\lambda)$) from ${\cal L}$ to
${\cal K}$. Naturally, {\em nobody\,} would ever use the complicated
3HS representation of unitary quantum systems if the conventional
textbook space ${\cal L}$ were {\em not\,} practically inaccessible.
Thus, the challenges emerging with the tests of stability appear
directly connected to the inaccessibility of space ${\cal L}$. A
quantification of the smallness of the perturbation does not seem
feasible without the necessary knowledge of the {\em perturbed\,}
version of the metric.

The ultimate conclusion is that we are only allowed to use the
${\cal K}-$space image of the inaccessible ${\cal L}-$space
perturbation,
 \be
 V_{{\lambda}} =
 \Omega^{-1}\mathfrak{v}_{{\lambda}}\, \Omega\,.
 \ee
This enables us to recall the knowledge of the unperturbed matric
$\Theta$ and the physical hidden-Hermiticity constraint
 \be
 V^\dagger_\lambda\,
 \Theta=\Theta\,V_\lambda\,.
 \label{27b}
 \ee
A relationship between operator $V_\lambda$ and the experimental
information carried by the upper-case perturbation $W_\lambda$ of
Eq.~(\ref{27}) is still far from clear. As long as it is mediated by
Eq.~(\ref{23}) it may be given the form
of a map from ${\cal K}$ to ${\cal L}$,
 \be
 (\mathfrak{h}+{{\lambda}} \mathfrak{v}_\lambda)\Omega
 (1+{{\lambda}}\,\Delta_\lambda)
 =\Omega(1+{{\lambda}}\,\Delta_\lambda)(H+{{\lambda}} W_\lambda)\,.
 \label{43}
 \ee
After its minor
re-arrangement we obtain the full-fledged upper-case space$-{\cal
K}$ version of the correspondence between $W_\lambda$ and
$V_\lambda$,
 \be
 (H+{{\lambda}}\, V_\lambda) (1+{{\lambda}}\,\Delta_\lambda)
 =(1+{{\lambda}}\,\Delta_\lambda)(H+{{\lambda}}\, W_\lambda)\,.
 \label{44}
 \ee
Such a relation confirms that $V_\lambda \neq W_\lambda$. A
simplification of this formula can yield the explicit definition of
the given, dynamical-input operator $W_\lambda$ in terms of
$V_\lambda$ (with the norm-boundedness still under our control,
in principle at least).
{\it Vice versa}, the reconstruction of $V_\lambda$ from the given
input $W_\lambda$ reads
 \be
 V_\lambda=
 (1+{{\lambda}}\,\Delta_\lambda)\,W_\lambda\,
 (1+{{\lambda}}\,\Delta_\lambda)^{-1}+
 (\Delta_\lambda\,H-H\,\Delta_\lambda)\,
 (1+{{\lambda}}\,\Delta_\lambda)^{-1}\,.
 \label{explici}
 \ee
This was the last step towards the following
important conclusion
concerning the operational admissibility of the Hamiltonian.

\begin{lemma}$\!\!\!.$
The insertion of operator $V_\lambda=V_\lambda(W_\lambda)$ of
Eq.~(\ref{explici}) in Eq.~(\ref{27b}) yields the operator-equation criterion
which tests the hidden Hermiticity
of the
perturbed Hamiltonian at any fixed
component $\Delta_\lambda$ of the physical Hilbert-space metric
$\Theta_\lambda$.
\label{unote}
\end{lemma}
In the scenario with
constant input $W \neq W({{\lambda}})$, formula (\ref{explici})
implies a manifest ${{\lambda}}-$dependence of output
$V=V({{\lambda}})$ (or of $\mathfrak{v}=\mathfrak{v}({{\lambda}})$).
The same observation would also hold in the opposite direction,
moving from $\mathfrak{v}\neq \mathfrak{v}({{\lambda}})$ and $V \neq
V({{\lambda}})$ to $W = W({{\lambda}})$. In both of these scenarios
one may conclude that
 \be
 V_\lambda-W_\lambda=\Delta_0\,H-H\,\Delta_0+{\cal O}({{\lambda}})\,.
 \label{keyeq}
 \ee
This means that the difference between $V_\lambda$ and $W_\lambda$
(caused by the Hamiltonian-dependence of the metric) is, in
general, model-dependent and non-perturbative, i.e., not necessarily
small at small ${{\lambda}}$. In fact, such an observation
exemplifies one of the differences between the present approach to
(general, non-Hermitian) perturbations and the traditional textbook
forms of perturbation theory.

\subsection{Perturbation expansions of metric\label{uuu}}

The feasibility of the constructions at
small $\lambda$ is based on the analyticity
assumptions and Taylor series
 \be
  W_{{\lambda}}=W_0
 +{{\lambda}}\,W^{(1)}
 +{{\lambda}}^2\,W^{(2)}+\ldots\,,
 \label{w25ab}
 \ee
 \be
  \Delta_{{\lambda}}=\Delta_0
 +{{\lambda}}\,\Delta^{(1)}
 +{{\lambda}}^2\,\Delta^{(2)}+\ldots\,.
 \label{25ab}
 \ee
Similarly, in Eq.~(\ref{26}) we expand
 \be
 T_{{\lambda}}=[1+{{\lambda}}\,\Delta^\dagger_{{\lambda}}]\,\Theta\,
 [1+{{\lambda}}\,\Delta_{{\lambda}}]
 =\Theta
 +{{\lambda}}\,T^{(1)}
 +{{\lambda}}^2\,T^{(2)}+\ldots
 \,
 \label{26exp}
 \ee
with
 \be
 T^{(1)}=\Delta_0^\dagger\,\Theta+\Theta\,\Delta_0\,,
 \ \ \ \
%
%
 T^{(2)}=(\Delta^{(1)})^\dagger\Theta
 +\Delta_0^\dagger\,\Theta\,\Delta_0+\Theta\,\Delta^{(1)}
 \label{hui}
 \ee
etc.
These series may be inserted in the
crypto-Hermiticity relation (\ref{27}) of Lemma \ref{te27}.
In dominant order the latter relation acquires the form
which is satisfied by assumption,
 $$
 H^\dagger \Theta=\Theta\,H\,.
 $$
In the spirit of perturbation theory
the validity of this relation between the two
given unperturbed
operators $H$ and $\Theta$ is granted.
We may turn attention to
the crypto-Hermiticity relation in the next-order approximation,
  $$
 H^\dagger\,T^{(1)}+W_0^\dagger \,\Theta = \Theta \,W_0+T^{(1)}\,H\,.
 $$
It is to be interpreted as an equation by which one converts the
known leading-order dynamical input (i.e., operator $W_0$) into the eligible
output information about the first order correction $T^{(1)}$ to the
(in general, non-unique) physical Hilbert-space metric.

In the next
step one can proceed to the second order relation
 $$
 H^\dagger\,T^{(2)}+W_0^\dagger \,T^{(1)} +(W^{(1)})^\dagger \,\Theta =
 \Theta \,(W^{(1)})
 +T^{(1)}\,W_0
 +T^{(2)}\,H\,.
 $$
It contains the
operator $T^{(1)}$ obtained in the preceding step.
The second-order correction $T^{(2)}$
to the metric is to be deduced.
Along the same lines one can also evaluate the higher-order corrections.

Alternatively, we might recall Eqs.~(\ref{hui}) and redefine our task
as the search for separate components of the Dyson map in
expansion (\ref{25ab}).
In the first-order approximation the
reconstruction of $\Delta_0$ from the leading-order
perturbation component $W_0$
has the form of equation
 \be
 \widetilde{W_0}^\dagger\,\Theta=\Theta\,\widetilde{W_0}\,,
 \ \ \ \ \widetilde{W_0}=W_0+\Delta_0\,H-H\,\Delta_0\,.
 \ee
In it we recognize relation (\ref{27b}). The abbreviation
$\widetilde{W_0}$ just represents perturbation $V_\lambda$ in the
leading-order approximation of Eq.~(\ref{keyeq}). The variable
$\lambda$ is absent, and the interpretation of the equation is now
approximative. Just the leading-order correction $\Delta_0$ is
deduced from the leading-order input $W_0$.

\section{Discussion\label{fifi}}


One of the characteristic, well-remembered \cite{Styer} features of
the birth of quantum mechanics may be seen in its more or less
parallel emergence in the two different (viz, in the Heisenberg's
\cite{Heisenberg} and in the Schr\"{o}dinger's \cite{Schroedinger})
alternative (albeit equivalent)
formulations. Quickly, these two phenomenologically complementary
formulations {\it alias\,} pictures
found
their specific areas of applicability \cite{Messiah}.
Step by step, they were also followed by a
number of further, still fully equivalent reformulations of the
theory, with remarkably different fields of
optimal implementation. {\it
Pars pro toto\,} let us mention the Feynman's
treatment of quantum
theory using path integrals which
opened new horizons in field theory \cite{Feynman}, or the
Bender's and Boettcher's
very recent recommendation
of a turn of attention to
complex interaction
potentials yielding still
the real, stable bound-state spectra \cite{BB}.

The Bender's and Boettcher's proposal
inspired an increase of
experimental as well as theoretical
activities in several different
(i.e., not necessarily just quantum)
branches of
physics \cite{book,Carlbook,Musslimani}.
In such a deeply innovative,
fairly broad and flexible
methodical framework
we tried to
fill here one of the gaps
in the theory by analyzing
a slightly enigmatic problem of the response of a generic
non-Hermitian system to perturbations.
Let us now complement our
mathematical results by a few
physics-oriented remarks.

\subsection{Physics of open {\it versus} closed quantum systems\label{opclo}}

During the study of
quantum
dynamics
controlled
by a
parameter-dependent Hamiltonian $H_\lambda$
which happens to be non-Hermitian in a pre-selected
and, in general,
Hamiltonian-dependent Hilbert space ${\cal K}_\lambda$,
one has to distinguish
between the open-system scenarios
(reflecting, e.g., the
decay of resonances)
and the description of evolution of a closed system
(characterized, first of all, by the unitarity, i.e., by
the conservation of the norm).
In
our present paper we paid attention to
the latter dynamical regime.
We were interested
in a mathematically
consistent description of
quantum systems and of
the changes of their observable properties (e.g., stability)
under perturbations.

In a broader context of applied physics
the emergence of the concept of non-Hermitian Hamiltonians can,
and should, be
traced back to two independent sources. The older one was
offered by Feshbach \cite{Feshbach}. It consists in a
reduction of the conventional Hilbert space of states ${\cal L}$ to
a ``model'' subspace ${\cal K}$.
The trick helped to simplify the language of the theory in
applications.
In our present paper only a marginal attention was paid to
such a traditional non-Hermitian open-system theory. Interested
readers may consult, e.g., monograph \cite{Nimrod} for its recent
review.

The basic idea of the study of
the other, so called closed quantum systems via non-Hermitian quantum
Hamiltonians
is
usually attributed to Dyson \cite{Dyson}.
In practice, the Feshbach- and Dyson-related model-building
strategies based on non-Hermitian Hamiltonians appeared applicable
in different situations.
In contrast to the latter, 3HS strategy,
its former, Feshbach-inspired counterpart deals with
unstable systems and with their resonant states. Their evolution is,
typically, controlled by a non-Hermitian effective Hamiltonian which
may be even energy-dependent \cite{enedep}. One might
refer to the open-system model-subspace theory (in which the
spectrum appears to be complex in general) under the acronym of
``two-Hilbert-space (2HS) formulation of quantum mechanics''.

Occasionally, the  2HS -- 3HS parallels may
be of interest. For example, space ${\cal K}$ must
remain, from the user's perspective, preferable to the
textbook space ${\cal L}$. In contrast to the 2HS
case, the auxiliary 3HS space ${\cal K}$
must always be perceived as unphysical.
An inner-product operator $\Theta$ of metric must be introduced
to reconvert the unphysical space ${\cal K}$ into its correct
alternative ${\cal H}$.

\subsection{The problem of stability\label{sxxx}}

A not quite expected parallel emerges between the 1HS
and 3HS criteria
of admissibility and smallness of perturbations.
In the conventional quantum theory
people usually expect that the stability of a system is lost
whenever the
perturbation $\lambda\,\mathfrak{v}(\lambda)$ ceases to be
self-adjoint. Similarly, our present analysis came to the formally
more complicated but analogous conclusion that
whenever we need to guarantee a closed-system stability
(i.e., the reality of spectrum),
the
hidden Hermiticity constraint has to be imposed
upon the perturbations
$\lambda\,W(\lambda)$.

For a deeper understanding of the latter parallel it is necessary to
recall that in the conventional Schr\"{o}dinger representation
the unitarity of
evolution is guaranteed by the self-adjointness of the Hamiltonian.
The
Dyson-inspired simplification of the Hamiltonian $\mathfrak{h} \to
H=\Omega^{-1}\,\mathfrak{h}\,\Omega$ has been found useful even when
accompanied by the loss of Hermiticity.
Such an isospectral preconditioning is to be read as
changing the Hilbert space into an auxiliary,
unphysical one, ${\cal L} \to {\cal K}$. A subsequent
re-Hermitization of $H$, i.e., the introduction of a suitable
inner-product metric operator $\Theta=\Theta(H)$ then reconverts
${\cal K}$ into the third Hilbert space ${\cal H}$ which is physical
again, i.e., unitarily equivalent to ${\cal L}$.
What is crucial is
that any perturbation  $H \to H_\lambda=H+\lambda\,W(\lambda)$ of
the Hamiltonian in ${\cal K}$ necessarily induces the changes of
geometry in the relevant physical Hilbert space, ${\cal H} \to {\cal
H}_\lambda$. One must suspect
that after the inclusion of perturbations the formalism might lose
its internal mathematical consistency as well as any acceptable
physical interpretation.

In our paper the mathematical
explanation of such an
apparent paradox was described.
Nevertheless,
any exhaustive discussion of
the physical aspects of the 5HS approach to the questions of stability
must necessarily be model-dependent. Certainly,
any guarantee of the
smallness of perturbations
must inseparably be related to the
phenomenology-determining connection between the metric
and a {\em complete\,}
set $\{\Lambda_j\}$ of the
non-Hermitian operators of observables
which is chosen as ``irreducible
in ${\cal K}$'' \cite{arabky}.

A more
detailed analysis of this topic may be found
in the physics-oriented
review \cite{Geyer}. The ambition of the authors
was to establish ``a general
criterion for a set of non-Hermitian operators to constitute a
consistent quantum mechanical system''.
One of the most successful
implementations of this criterion
was discovered and made popular by Bender with coauthors
who studied multiple specific ordinary differential
crypto-Hermitian
quantum models
exhibiting additional antilinear
symmetries of $H= -d^2/dx^2+V(x)$ called
${\cal PT}-$symmetry and
${\cal PCT}-$symmetry (see, e.g., \cite{Carl,ali} for details).

\subsection{The requirement of unbroken ${\cal PT}-$symmetry}

The abstract 3HS formulation of quantum mechanics may be traced back
to Dieudonn\'{e} \cite{Dieudonne}. He pointed out, as early as in
1961, that many of the rigorous mathematical
aspects of the 3HS theory are rather discouraging in the general case.
Later, a number
of further critical arguments has been found against the use of
unbounded Hamiltonians $H$ (cf., e.g.,
\cite{Viola}). A partial remedy has only been
found in review \cite{Geyer} dated 1992. There, the authors
circumvented some of the most serious formal drawbacks
of the general 3HS formalism by admitting
{\em only\,} such ``quasi-Hermitian'' observables (including Hamiltonians
$H$) which are
{\em bounded\,} in ${\cal K}$.

A return of physicists to unbounded
non-Hermitian Hamiltonians with real spectra appeared motivated,
a few years later, by
the needs of quantum field theory \cite{Milton}. This inspired an
intensive search for suitable toy models. In particular, for
one-dimensional Schr\"{o}dinger operators
\cite{BB,HO,Geza} a serendipitous benefit appeared to emerge from an
exceptionally straightforward realization of the Hermitization
${\cal K}\to {\cal H}$ \cite{Carl}. The goal has been achieved
thanks to a fortunate combination of the ${\cal PT}-$symmetry
assumption
 \be
 H{\cal PT}  ={\cal PT} H\,
 \label{ptsy}
 \ee
(where ${\cal P}$ is parity and ${\cal T}$ represents the time
reversal) with the postulate of  existence of an {\it ad hoc\,}
operator ${\cal C}$ exhibiting properties of a charge \cite{Carl}.
This facilitated the construction of the metric because the
Dieudonn\'{e}'s relations (\ref{cryher}) appeared satisfied by the
product $\Theta_{\cal PT}={\cal PC}$ \cite{Carl,Carlbook}.
Thus,
in a widely accepted terminology
referring to the original
motivating context of quantum field theory
people usually connect the EP-related
instant of the loss of the reality of the spectrum of $H$
with the loss of the ${\cal PT}-$symmetry of the wave functions,
giving the loss the
name of the ``spontaneous breakdown of ${\cal PT}-$symmetry''
\cite{Carl}.

After 2012, several concrete implementations
of
the model-building strategy
based on the antilinear
symmetries (\ref{ptsy})
revealed that
the rigorous mathematical background of such an approach
may happen to be perceivably more complicated than
originally expected
(cf., e.g., \cite{ATbook,Uwe,Baga,Petr}).


\subsection{Physical meaning of the reality of the energies}

In the majority of the traditional textbooks on quantum mechanics
\cite{Messiah} the exposition of the theory usually starts from the
description of the closed systems, i.e., from the conventional
bound-state Schr\"{o}dinger equation which is, say,
real-parameter-dependent.
The evolution of these systems in time is unitary and, in the light
of Stone theorem \cite{Stone}, the Hamiltonian itself is
self-adjoint in the corresponding Hilbert space of states ${\cal
L}$.

Whenever people turn attention to the so called open quantum
systems the self-adjointness constraint must be omitted because the
spectrum of the energies need not be real anymore. Still,
along strictly the same lines,
one can also consider a number of more general,
non-Hermitian resonant-state Schr\"{o}dinger equations
 \be
 \widehat{H(\lambda)}\,|\psi(\lambda) \kt =
 \widehat{E(\lambda)}\,|\psi(\lambda)
 \kt\,
 \label{papertse}
 \ee
with non-real spectra $\{\widehat{E(\lambda)}\}$.
The
search for the solutions
need not be followed by any reconstruction of the metric
because
the topological vector space ${\cal K}$
of the open-system states $
|\psi(\lambda) \kt\ \in \ {\cal K}$ becomes physical \cite{Muga}.

The formal parallels between the closed and open quantum systems
become even stronger when one is allowed to assume that the
$\lambda-$dependence of the solutions remains weak. Then
one usually treats the Hamiltonian as
composed of a reference operator and a small perturbation,
 \be
 \mathfrak{h}_\lambda=\mathfrak{h}+\lambda\,\mathfrak{v}
 =\mathfrak{h}_\lambda^\dagger\,,
 \ \ \ \ \
 \widehat{H_\lambda}=\widehat{H}+\lambda\,\widehat{V}
  \,\neq \,\widehat{H_\lambda}^\dagger\,.
 \label{urere}
 \ee
It is not too difficult to proceed in a consequent
parallel with the self-adjoint case. With obvious modifications: the
conventional unperturbed bases must be replaced by their
biorthonormalized generalizations. As long as $\widehat{H_\lambda} \neq
\widehat{H_\lambda}^\dagger$, one has to consider Schr\"{o}dinger
equations for the respective left and right eigenvectors of the
Hamiltonian,
 \be
 [\widehat{H_0}+\lambda\,\widehat{V}]\,|\psi(\lambda) \kt =
 \widehat{E}(\lambda)\,|\psi(\lambda)
 \kt\,,
 \ \ \ \
 |\psi(\lambda) \kt\ \in \ {{\cal K}_\lambda}\,,
 \label{apertse}
 \ee
 \be
 [\widehat{H_0}^\dagger+\lambda\,\widehat{V}^\dagger]\,|\psi(\lambda) \kkt =
 \widehat{E}^*(\lambda)\,|\psi(\lambda) \kkt\,,
 \ \ \ \
 |\psi(\lambda) \kkt\ \in \ {\cal K}_\lambda\,.
 \label{bepertse}
 \ee
The energies need not be real in general (so that the superscripted
star $^*$ marks here the complex conjugation) but under suitable
mathematical assumptions they can still be sought via the same
Rayleigh-Schr\"{o}dinger ansatz \cite{Kato},
 \be
 \widehat{E}(\lambda)=\widehat{E}(0) + \lambda\,\widehat{E_1}
 +\lambda^2\,\widehat{E_2}+\ldots\,.
 \label{20hat}
 \ee
One of the key distinctive features of the hatted energies
(\ref{20hat}) is that they are, in
general, complex.
For this reason, as we already emphasized, we
were only interested here in the ``unhatted'' Hamiltonians $H_\lambda$
possessing the {\em strictly real\,} spectra and exhibiting the {\em
additional\,} hidden-Hermiticity property (\ref{dudu}).

\subsection{The concept of admissible perturbations\label{usxxx}}

On the basis of physical-reality-reflecting
requirements one always has to specify clearly the class of the
experimentally realizable perturbations, irrespectively of the
formalism. Such a specification seems missing in the literature
on 3HS models,
so in our recent comment \cite{Ruzicka} and
paper \cite{admissible} we
turned attention to
the topic. We pointed out that the variations of the Hamiltonian
cause also the variations of
the geometry of the underlying physical Hilbert space.
We found the consequences ``strongly counterintuitive''.
In the present continuation of these considerations we
offered one of possible
resolutions of the puzzle, reflecting
the well known danger
\cite{stellenbosch}
that the presence of exceptional points
at small $\lambda$
could endanger the convergence
of Rayleigh-Schr\"{o}dinger series.
For this reason
the acceptability of our present approach
is based on the assumption that these singularities
stay sufficiently remote
so that
the consequent perturbation-approximation
strategy remains admissible and mathematically well founded.

Let us point out that the Hilbert-space geometry
(controlled by the physical inner-product metric $\Theta$) can
acquire its $\lambda-$dependence in two complementary ways, viz,
directly and indirectly. Besides the obvious, ``direct''
$\lambda-$dependence inherited from the Hamiltonian one must take
into account that the assignment $H \to \Theta(H)$ remains, at every
$\lambda$, ambiguous. Thus, indirectly, the removal of
this ambiguity also remains at our disposal as a phenomenologically
relevant {\em independent\,} degree of freedom in model-building
\cite{Geyer,lotor}.
The latter observation was a core of our present construction of the
general re-Hermitization scheme admitting the
analysis of the admissibility of the perturbations.

In diagram (\ref{findiag}) we showed, in particular, that it
is necessary to distinguish between as many as five relevant Hilbert
spaces, in principle at least.
In practice, the pattern admits
various simplifications. We succeeded in reducing
the set of the relevant spaces to three (cf. diagram
(\ref{dvacetosm})).
This renders the explicit perturbation-expansion constructions
of physical states possible. These constructions
are expected to be
based again on the most natural power-series ansatzs
sampled by Eq.~(\ref{20hat}).
Thus, our technique of description of the perturbed
unitary quantum systems do not seem to lead to any
inconsistencies.

\subsection{Outlook}

Any
realistic application of our perturbation-expansion recipe
will be strongly model-dependent so that it remained
beyond the scope of our present paper. We only paid detailed
attention to a few key technicalities. Under the assumption that the
$\lambda-$dependence of our Hamiltonians $H_\lambda$ happens to be
analytic we managed to show that the extension of the conventional
perturbation-expansion techniques to the ${\cal PT}-$symmetry
context is, up to a few important differences, straightforward. One
of the main explicit technical messages delivered by our present
paper is that whenever the unitary, 3HS-represented quantum system
with the Hamiltonian $H$ which is non-self-adjoint in ${\cal K}$ but
self-adjoint in ${\cal H}$ gets exposed to a perturbation
$\lambda\,W$ which is non-self-adjoint in ${\cal K}$,
there exists an appropriate
generalization of the textbook perturbation theory which may be
constructed in ${\cal K}$ and which gives the correct physical
predictions.

We clarified that the traditional notions of the
norm or of the smallness of perturbations in ${\cal L}$ cannot be
too easily transferred to the 3HS framework. Whenever the
conventional perturbed Hamiltonians $\mathfrak{h}_{\lambda}
=\mathfrak{h} + {\lambda} \mathfrak{v}_{\lambda}$ are not found
prohibitively complicated, the analysis of perturbations should
certainly be performed in the textbook space ${\cal L}$. Only when
this is not the case, a transition to the quasi-Hermitian picture
becomes a viable strategy.
Indeed, even though the non-Hermitian representation approach
proves reasonably straightforward and mathematically well founded,
it still remains not too easy to implement.

Our key message
concerns the criteria of the smallness of
perturbations.
In contrast to the conventional
Hermitian perturbation theories \cite{Messiah},
and in contrast to the
studies of the response to perturbations
in open quantum systems \cite{Viola},
a {\em decisive\,} component of {\em any\,}
mathematically
consistent perturbation-approximation description of {\em any\,}
closed crypto-Hermitian model lies in the
reconstruction of
the physical Hilbert-space metric $\Theta_\lambda$.
One encounters here
a hidden parallel with the Hermitian theories:
The ``operationally admissible''
perturbations $\lambda\,W(\lambda)$
are only those which are
self-adjoint in
the physical Hilbert space.
As long as the physical norm
(i.e.,
the metric)
necessarily {\em varies\,} with
the parameter $\lambda$ in general,
one {\em cannot\,}
confirm or disprove the smallness of a given
quasi-Hermitian perturbation $\lambda\,W(\lambda)$
too easily
(interested readers might
find a fairly discouraging illustrative example
in \cite{corridors}).

For the practical purposes this implies that
as long as the construction of $\Theta_\lambda$
remains an inseparable part of the 5HS formalism,
the
questions of stability of a given system with respect to random
perturbations (answered, in open systems, by the construction of the
pseudospectrum \cite{Viola,Trefethen}) do not seem to have any easy
or exhaustive
answer for the closed crypto-Hermitian quantum systems at present.
Mostly, only an {\it a posteriori}, self-consistent test of the
smallness of perturbation remains at our disposal.

A certain consolation may only be sought in the fact
that this is, after all,
a situation which is not too different
from the situation in the single-Hamiltonian setting.
Indeed,
the 5HS ambiguity of
the specification of the  ``admissible''
physical $\lambda \neq 0$ Hilbert space ${\cal H}_\lambda$
is merely inherited from
the unperturbed
3HS formalism where $\lambda=0$.
It can {\em always\,} be suppressed,
in principle at least, via
an
extended dynamical input
access to
a {\em complete\,}
set $\{\Lambda_j\}$ of the
non-Hermitian operators of observables
in ${\cal K}$.
Fortunately, citing the words of review
\cite{Geyer},
``it is not always necessary to construct the metric
for the whole set of
observables under consideration'' in applications.

\section{Summary\label{susu}}

In multiple applications
the Rayleigh-Schr\"{o}dinger perturbation series is
found convergent. According to the Kato's book \cite{Kato}, the
radius of the corresponding circle of
convergence $\lambda_{\max}$ is equal to the distance of
$\lambda=0$ from the (in general, complex)
set of the so called exceptional-point (EP) values
$\lambda^{(EP)}$. In this way the question about the ``smallness''
of perturbations $\lambda \mathfrak{v}$
is given a more or less exhaustive
abstract answer.

This answer appears not too constructive even in
Hermitian theory
where the EP values are never real.
This leads to the necessity of working
with the analytically continued Hamiltonians
which cease to be self-adjoint.
The less satisfactory the answer is when
we admit, in
a preferred Hilbert space ${\cal K}$,
a manifest non-Hermiticty of the preselected
perturbed
Hamiltonian $H_\lambda$
for which the energy spectrum may but need not be real.
In this sense, our present
study of crypto-Hermitian $H_\lambda$s
may be summarized as reopening the problem
of the smallness of perturbations,
and as its constructive and mathematically consistent reformulation.

Our starting point was the standard
3HS theory, the probabilistic
interpretation of which
is known to be
obtained via
a positive definite, Hamiltonian-dependent
physical
Hilbert-space metric.
We observed that once one admits
the $\lambda-$parametric variability
of the Hamiltonian $H_\lambda$,
the task of making the
physical contents of the theory consistent
becomes nontrivial even when the variability of $H_\lambda$
is only
caused by a ``small'' perturbation in $H_\lambda=H+\lambda W$
(in the unperturbed
limit
$\lambda \to 0$
we dropped the usual zero subscript as redundant).

A key theoretical obstacle
was found in the
transfer of the parameter-dependence from $H_\lambda$
to
all of the components of the formalism and, in particular,
to
the Dyson's map ($\Omega \to \Omega_\lambda$).
At length we discussed the resulting
parameter-dependence modifications of
the mathematical manipulation space (${\cal K} \to {\cal K}_\lambda$)
as well as of
the ultimate physics-representing
space (${\cal H} \to {\cal H}_\lambda$).
Thus, we had to
apply the 3HS formalism twice, viz., to the
unperturbed
Hamiltonian
$H$
and, separately, to its perturbed partner $H_\lambda$ associated with
the modified,
$\lambda-$subscripted Hilbert-space triplet
$\{{\cal L}_\lambda,{\cal K}_\lambda, {\cal H}_\lambda\}$.
Finally,
without any loss of generality we
postulated the $\lambda-$independence of the
textbook Hilbert spaces (${\cal L}_\lambda={\cal L}$).
This enabled us to
merge the two
respective
Dyson-inspired 3HS formulations of quantum mechanics
of Ref.~\cite{Geyer}
into a single 5HS flowchart.

It appeared not too surprising that in the resulting
constructive
recipe
a number of subtle
technical as well as physical interpretation problems emerged
(one of them, with $\lambda_{\max}=0$,
was recently reported in \cite{corridors}).
For this reason, just
a smooth jump from $\lambda=0$ to $\lambda \neq 0$
was considered, and
some of the consequences were discussed in detail.
On this background it may be concluded that the
``non-Hermitian'' (or, more precisely, ``quasi-Hermitian``
\cite{Geyer} or ``crypto-Hermitian'' \cite{Smilga}) 5HS
formulation of perturbation theory remains nontrivial but
still feasible and reasonably transparent and user-friendly.

\end{document}